\documentclass[aps,pra,10pt,twocolumn,superscriptaddress,nofootinbib,floatfix]{revtex4-2}
\usepackage[english]{babel}
\usepackage[utf8]{inputenc}
\usepackage[T1]{fontenc}
\usepackage{lmodern}
\usepackage{amsmath,amssymb,amsfonts,amsthm,bm,mathrsfs}
\usepackage{mathtools}
\usepackage{graphicx}
\usepackage{xcolor}
\usepackage[colorlinks=true,linkcolor=blue,citecolor=blue,urlcolor=blue]{hyperref}
\hypersetup{pdftitle={Symmetry-Enforced Quadratic Approximate-Degradability Bounds for Noisy Landau-Streater Channels},pdfauthor={Yun-Feng Lo, Yen-Chi Lee, Min-Hsiu Hsieh}}
\newtheorem{theorem}{Theorem}

\newtheorem{proposition}{Proposition}
\newtheorem{corollary}{Corollary}

\newcommand{\M}{\mathscr M}
\newcommand{\N}{\mathcal N}
\newcommand{\D}{\mathcal D}
\newcommand{\LL}{\mathcal L}
\newcommand{\id}{\operatorname{id}}

\newcommand{\hb}{h_2}
\DeclareMathOperator{\tr}{tr}
\begin{document}
\title{Symmetry-Enforced Quadratic Approximate-Degradability Bounds for Noisy Landau--Streater Channels}
\author{Yun-Feng Lo}
\affiliation{Department of Electrical and Computer Engineering, Georgia Institute of Technology, Atlanta, Georgia, USA}
\author{Yen-Chi Lee}
\affiliation{Department of Mathematics, National Central University, Taoyuan, Taiwan}
\author{Min-Hsiu Hsieh}
\affiliation{Quantum Computing Research Center, Hon Hai Research Institute, Taipei, Taiwan}
\date{2 October 2026}
\begin{abstract}
Understanding the limits of quantum communication is
a central problem in quantum information theory. Structural properties
of quantum channels provide tools for analyzing these limits.
In particular, approximate degradability links quantum-capacity
bounds to the accuracy with which a channel's environment can be
reconstructed from its output.
We study this reconstruction problem for noisy Landau--Streater
channels across finite-spin representations. For spin $j$ and noise
probability $p$, we construct an explicit degrading map with
diamond-norm error at most
$[2+8/(j(j+1))]p^2\le(38/3)p^2$ on the common interval
$0\le p\le1/8$. A lower bound establishes quadratic order for this
construction at each fixed spin. Floating-point semidefinite
optimization places its error within $2.48\%$ of the numerical
optimum on the sampled grid for four spins. Covariance and
conditional-entropy concavity yield explicit two-sided
quantum-capacity bounds for this family with a spin-uniform
remainder. Building on the low-noise Pauli cancellation strategy,
we formulate a sufficient condition based on a common scalar
adjoint action on fixed, possibly nonunitary noise Kraus operators.
The resulting reconstruction-error bound is uniform in system
dimension and Kraus count. The criterion applies to irreducible
$SU(N)$ generator channels and noisy Werner--Holevo channels,
with the $SU(3)$ adjoint representation providing an example
outside the mixed-unitary class.
\end{abstract}
\maketitle
\begingroup
\renewcommand{\thefootnote}{}
\footnotetext{
This work was supported by the National Science and Technology Council of Taiwan under Grant No.~NSTC 113-2115-M-008-013-MY3.\\
Corresponding author: Yen-Chi Lee (\href{mailto:yclee@math.ncu.edu.tw}{yclee@math.ncu.edu.tw}).\\
Y.-F. Lo and Y.-C. Lee contributed equally to this work.}
\addtocounter{footnote}{-1}
\endgroup

\section{Introduction}\label{sec:introduction}

Covariant quantum channels provide a setting in which symmetry
can simplify quantum-information tasks \cite{holevo2012quantum}.
The Landau--Streater channel is a rotationally covariant example
built from angular-momentum generators in a finite-spin
representation \cite{landau1993birkhoff,filippov2019quantum}.
When this noise acts with a small probability, increasing the
spin enlarges the information-carrying system while retaining
three noise generators. This family offers a concrete setting
for studying how the accuracy of environmental reconstruction
depends on system size. We construct a degrading map whose
reconstruction error admits an explicit quadratic bound on a
noise interval common to all spins.

Such a reconstruction guarantee helps bound the quantum
capacity, which quantifies the asymptotic rate of reliable
quantum communication. This capacity is generally expressed
as a regularized optimization of coherent information
\cite{bennett1998quantum,wilde2017quantum,holevo2012quantum}.
Because coherent information can be superadditive, its optimized
single-use value need not determine the capacity, making evaluation
difficult even for comparatively simple noise models
\cite{divincenzo1998quantum,cubitt2015unbounded,leditzky2018dephrasure}.
For degradable channels, however, coherent information is additive,
so the regularized expression reduces to a single-letter capacity
formula \cite{devetak2005capacity,cubitt2008structure}.
This simplification motivates approximate degradability, which
relaxes the exact reconstruction requirement: one seeks a channel
that reconstructs the complementary output from the receiver's
output up to a small diamond-norm error
\cite{sutter2017approximate}.
Controlling this reconstruction error then yields quantitative
bounds on quantum capacity and on the possible advantage of
regularization.

When a noisy channel is close to the identity, the receiver's
output remains close to the input, while the environment may
still acquire information about that input.
This raises the question of how accurately the environment's
state can be reconstructed from the receiver's output as the
noise strength decreases.
Leditzky, Leung, and Smith established a generic
$O(\epsilon^{3/2})$ approximate-degradability bound for channels
$\epsilon$-close to the identity and sharper quadratic bounds
for qubit Pauli channels \cite{leditzky2018prl}.
Their Pauli-channel construction achieves
this improvement by adjusting the noise parameters of the
degrading map to cancel the leading $O(p^{3/2})$ contributions
to the reconstruction error, where $p$ parameterizes the noise.
They also note that the argument extends to higher-dimensional
generalized Pauli channels through their diagonal action in
the Weyl basis, but do not provide an explicit dimension-uniform
quadratic bound. To obtain a guarantee that remains uniform
as the system dimension increases, we build on this cancellation
strategy and formulate a checkable condition on how the adjoint
noise channel acts on its fixed, possibly nonunitary Kraus
operators. Under this condition, we derive explicit quadratic
error bounds whose constants and validity interval are
independent of both system dimension and Kraus
count.

The noisy Landau--Streater (NLS) family introduced above
illustrates how this condition applies beyond unitary noise
operators. The channel has dimension $d=2j+1$ and reduces
to the qubit depolarizing channel at $j=1/2$; for $j>1/2$,
its spin-generator noise operators are not proportional
to unitaries. For this family, a Casimir identity verifies
the required common adjoint-action condition, so the
cancellation mechanism can be checked algebraically.
Related noisy Landau--Streater and Werner--Holevo families
have been studied by Roofeh and Karimipour
\cite{roofeh2024noisy,karimipour2024noisy};
Lie-algebra channel constructions more generally were studied
by Ritter \cite{ritter2005quantum}.

Our main contributions are as follows.

\textit{(i) A spin-uniform degrading-map bound.}
Writing $\eta_j(p)$ for the optimal degradability error and
$e_j(p)$ for the error of our explicit construction, we prove
\begin{equation}\label{eq:intro-bound}
\begin{aligned}
 \eta_j(p)\le e_j(p)
 &\le \left(2+\frac{8}{j(j+1)}\right)p^2\\
 &\le \frac{38}{3}p^2,
 \qquad 0\le p\le\frac18.
\end{aligned}
\end{equation}
The lower bound $e_j(p)\ge4p^2/[j(j+1)]$ establishes quadratic
order for this construction at every fixed spin.
We also show that a common mixing probability $p$ gives all
spins the same diamond distance $2p$ from the identity,
providing an exact calibration for the comparison.

\textit{(ii) An adjoint-action criterion.}
We show that a common scalar adjoint action on the fixed noise
Kraus operators is sufficient for quadratic approximate
degradability. The resulting bound is $18p^2$ on
$0\le p\le1/8$, independent of system dimension and Kraus
count. We verify the condition for irreducible $SU(N)$
generator channels and noisy Werner--Holevo channels,
including an $SU(3)$ adjoint example outside the
mixed-unitary class.

\textit{(iii) A numerical comparison with optimal degrading maps.}
We compare the explicit construction with unrestricted
semidefinite optimization over degrading maps.
On the sampled grid for four spins, the construction's error
is within $2.48\%$ of the numerical optimum, as assessed
using repaired floating-point SDP witnesses.

\textit{(iv) Explicit capacity bounds.}
We use covariance and conditional-entropy concavity to
evaluate the conditional-entropy term in the
approximate-degradability upper bound.
Combined with the coherent information of the maximally
mixed input, this gives explicit two-sided capacity bounds
with a spin-uniform remainder.

The remainder of this paper is organized as follows.
Section~\ref{sec:prelim} defines the model and error conventions.
Sections~\ref{sec:geometry} and \ref{sec:main_result} establish
the noise calibration and uniform constructive bound.
Section~\ref{sec:capacity} evaluates the capacity consequence,
and Sec.~\ref{sec:numerics} presents the numerical degrading-map
comparison. The general criterion and further channel families
are developed in Sec.~\ref{sec:mechanism}.
Appendices~\ref{app:supplementary-taylor}--\ref{app:supplementary-covariance}
provide optional pedagogical calculations; the main theorems and their
proofs are complete without these supplementary appendices.

%%----

\section{Preliminaries and channel model}\label{sec:prelim}
\subsection{Norms, complementary channels, and error conventions}
All systems are finite dimensional, and all logarithms are to base two. For an operator $X$, let $\|X\|_1=\tr\sqrt{X^\dagger X}$. The diamond norm of a linear map $\Phi:A\to B$ is
\begin{equation}\label{eq:diamond-definition}
 \|\Phi\|_\diamond=\sup_{\|X\|_1=1}
 \|(\id_A\otimes\Phi)(X)\|_1,
\end{equation}
where $X$ acts on an ancilla of dimension $\dim A$ and the input \cite{khatri2020principles}. This norm uses no factor of one half. Every completely positive trace-preserving (CPTP) map has diamond norm one, and the distance between two channels is at most two.

For Kraus operators $K_\mu$ with $\sum_\mu K_\mu^\dagger K_\mu=I$, we fix the complementary-channel convention
\begin{equation}\label{eq:complement-convention}
 \N^c(X)=\bigl[\tr(K_\mu X K_\nu^\dagger)\bigr]_{\mu,\nu}.
\end{equation}
A channel is degradable if $\N^c=\D\circ\N$ for some CPTP map $\D$ from its output to the environment. Its optimal approximate-degradability error is
\begin{equation}\label{eq:eta-general}
 \eta(\N)=\inf_{\D\ {\rm CPTP}}
       \|\N^c-\D\circ\N\|_\diamond.
\end{equation}
Any specified $\D$ supplies an upper bound on this infimum. Approximate-degradability capacity estimates depend on the error, the environment dimension, and entropy-continuity terms \cite{sutter2017approximate}; the form used here is stated in Sec.~\ref{sec:capacity}.

For numerical calculations we use unnormalized, input--output ordered Choi matrices,
\begin{equation}\label{eq:choi-convention}
 J(\Phi)=\sum_{a,b}|a\rangle\langle b|\otimes
                  \Phi(|a\rangle\langle b|).
\end{equation}
A map $\D:B\to E$ is CPTP exactly when $J(\D)\succeq0$ and $\tr_EJ(\D)=I_B$.

\subsection{Spin generators and the NLS family}
Let $j\in\{1/2,1,3/2,\ldots\}$, $d=2j+1$, and $\kappa=j(j+1)$. The Hermitian generators $J_1,J_2,J_3$ of the irreducible spin-$j$ representation satisfy
\begin{equation}\label{eq:spin-relations}
 [J_a,J_b]=i\sum_c\varepsilon_{abc}J_c,
 \qquad \sum_{k=1}^3J_k^2=\kappa I.
\end{equation}
We set $\hbar=1$. In the standard basis $|j,m\rangle$, $J_3|j,m\rangle=m|j,m\rangle$ and
\begin{equation}\label{eq:ladder}
 J_\pm|j,m\rangle=
 \sqrt{(j\mp m)(j\pm m+1)}\,|j,m\pm1\rangle,
\end{equation}
where $J_\pm=J_1\pm iJ_2$. The generators obey
\begin{equation}\label{eq:spin-traces}
 \tr J_k=0,\qquad
 \tr(J_kJ_l)=\frac{d\kappa}{3}\delta_{kl}.
\end{equation}

With $L_k=J_k/\sqrt\kappa$, define
\begin{align}
 \LL_j(X)&=\sum_{k=1}^3L_kXL_k,\label{eq:LS-def}\\
 \M_{j,p}(X)&=(1-p)X+p\LL_j(X),\quad 0\le p\le1.
 \label{eq:NLS-def}
\end{align}
The Casimir relation gives $\sum_kL_k^2=I$, so both maps are CPTP, unital, and self-adjoint. The Kraus list for $\M_{j,p}$ is
\begin{equation}\label{eq:NLS-kraus}
 K_0=\sqrt{1-p}I,\qquad K_k=\sqrt pL_k\quad(k=1,2,3).
\end{equation}
It fixes an environment of dimension four, independent of $j$.

For an arbitrary input operator $X$, the complement has the accessible block form
\begin{equation}\label{eq:NLS-complement-blocks}
 \M_{j,p}^c(X)=
 \begin{pmatrix}
 (1-p)\tr X&\sqrt{p(1-p)}\,\boldsymbol r(X)\\
 \sqrt{p(1-p)}\,\boldsymbol c(X)&p\LL_j^c(X)
 \end{pmatrix},
\end{equation}
where $r_k(X)=\tr(XL_k)$, $c_k(X)=\tr(L_kX)$, and
$[\LL_j^c(X)]_{kl}=\tr(L_kXL_l)$. The row and column coefficients agree by cyclicity of trace. For a non-Hermitian $X$ no complex conjugation of one into the other is implied.

Under rotations, the generators transform by an orthogonal spin-one matrix. Consequently,
\begin{equation}\label{eq:NLS-covariance}
 \M_{j,p}(U_gXU_g^\dagger)=U_g\M_{j,p}(X)U_g^\dagger.
\end{equation}
Indeed, inserting $U_g^\dagger J_kU_g=\sum_lR_{kl}(g)J_l$ into the Kraus sum and using $R^{\mathsf T}R=I$ leaves the sum unchanged. The Kraus span consists of a scalar and a vector, so the chosen complement is covariant with environment representation $1\oplus\text{spin-}1$. This covariance will be used for the capacity optimization.

\section{Exact diamond distance from the identity}\label{sec:geometry}
We first calibrate the mixing probability by the diamond distance
from the identity, providing a common noise scale for comparing
different spins.

\begin{theorem}[Distance from the identity]\label{thm:diamond_distance}
For every spin $j$,
\begin{equation}\label{eq:noise-distance}
 \|\LL_j-\id\|_\diamond=2,
 \qquad \|\M_{j,p}-\id\|_\diamond=2p.
\end{equation}
\end{theorem}
\begin{proof}
The distance between channels is at most two. For the reverse inequality, take a normalized maximally entangled state $|\Omega\rangle$ on two copies of the spin space. For every $k$,
\begin{equation}
 \langle\Omega|(I\otimes L_k)|\Omega\rangle
       =\frac{\tr L_k}{d}=0.
\end{equation}
Thus $(\id\otimes\LL_j)(|\Omega\rangle\langle\Omega|)$ has support orthogonal to $|\Omega\rangle$. Both states have trace one, so their trace-norm distance is two. Homogeneity then gives the second identity.
\end{proof}

This proof applies to any noise channel with traceless Kraus operators. Tracelessness thus provides the noise calibration, while the
adjoint action on the Kraus operators supplies the cancellation
mechanism developed next.

\section{A uniform constructive degradability bound}\label{sec:main_result}
We distinguish the optimized error from the error of a specified map:
\begin{align}
 \eta_j(p)&=\eta(\M_{j,p}),\label{eq:eta-NLS}\\
 e_j(p)&=\|\M_{j,p}^c-\M_{j,s}^c\circ\M_{j,p}\|_\diamond,
 \quad s=p+\frac{2p^2}{\kappa}.
 \label{eq:constructed-error}
\end{align}
The environment basis in Eq.~\eqref{eq:complement-convention} is fixed at both parameters. The ansatz $\D=\M_{j,s}^c$ follows the parameter-shift construction of Ref.~\cite{leditzky2018prl}.

\subsection{The adjoint action and cross-term cancellation}
The double-commutator identity for a vector operator is
$\sum_l[J_l,[J_l,J_k]]=2J_k$. Expanding it with the Casimir relation yields
\begin{equation}\label{eq:spin-triple}
 \sum_lJ_lJ_kJ_l=(\kappa-1)J_k,
 \qquad \M_{j,p}^\dagger(L_k)=(1-gp)L_k,
\end{equation}
where $g=1/\kappa$. The same scalar eigenvalue acts on all three noise operators.

The channel therefore attenuates all three spin moments by the same
factor. This common attenuation suggests compensating it by a single
increase in the complementary-noise parameter; the following expansion
determines the required shift.

For a shift $s=p+ap^2$, the cross-block coefficient of
$\M_{j,p}^c-\M_{j,s}^c\circ\M_{j,p}$ is
\begin{align}
 f(p,s)&=\sqrt{p(1-p)}-(1-gp)\sqrt{s(1-s)}\nonumber\\
       &=\left(g-\frac a2\right)p^{3/2}+O(p^{5/2}).
 \label{eq:NLS-cross}
\end{align}
Choosing $a=2g$ cancels the $p^{3/2}$ term.

The true and reconstructed environments retain different weights,
$(1-p,p)$ and $(1-s,s)$, in their no-error and error sectors.
This mismatch supplies the lower bound for the specified construction
in the theorem below.

The next theorem separates three quantities: $\eta_j(p)$ is the optimum over degrading maps, $e_j(p)$ is the error of the chosen map, and $\overline\varepsilon_j(p)$ is its analytical upper bound.

\begin{theorem}[Spin-uniform constructive bound]\label{thm:main}
For every spin $j$ and $0\le p\le1/8$, the parameter in Eq.~\eqref{eq:constructed-error} obeys $s\le1/6$. The corresponding CPTP degrading map satisfies
\begin{align}
 \eta_j(p)&\le e_j(p)\le\overline\varepsilon_j(p),\label{eq:uniform-NLS}\\
 \overline\varepsilon_j(p)&:=\left(2+\frac8\kappa\right)p^2
                   \le\frac{38}{3}p^2,\label{eq:epsilon-NLS}\\
 e_j(p)&\ge\frac{4p^2}{\kappa}.\label{eq:fixed-lower}
\end{align}
Consequently $e_j(p)=\Theta_j(p^2)$ as $p\downarrow0$, while these estimates establish only an $O(p^2)$ upper bound on $\eta_j(p)$.
\end{theorem}

\subsection{A diamond-norm block estimate}
\begin{proof}[Proof of Theorem~\ref{thm:main}]
Since $g\le4/3$, $s=p+2gp^2\le1/6$ on the stated interval. Put $\delta=s-p$ and $\N_p=\M_{j,p}$. Equation~\eqref{eq:NLS-complement-blocks} and the adjoint identity split the discrepancy into a scalar block, a noise block, and two cross blocks. The scalar-block map is $X\mapsto\delta\tr X$, with diamond norm $\delta$. The noise block is
\begin{align}
 p\LL_j^c-s\LL_j^c\N_p
 &=p\LL_j^c(\id-\N_p)-\delta\LL_j^c\N_p,
 \label{eq:noise-block-split}
\end{align}
whose norm is at most $2p^2+\delta$, since channels have diamond norm one and $\|\id-\N_p\|_\diamond\le2p$.

The two cross blocks equal $f(p,s)\mathcal R(X)$, where
\begin{align}
 \mathcal R(X)&=\tr_S(V_0XV_1^\dagger+V_1XV_0^\dagger),\nonumber\\
 V_0x&=x\otimes|0\rangle,\qquad
 V_1x=\sum_{k=1}^3L_kx\otimes|k\rangle.
 \label{eq:cross-isometries}
\end{align}
Both $V_0$ and $V_1$ are isometries because $\sum_kL_k^\dagger L_k=I$. Left and right multiplication by these isometries preserve trace norm; partial trace has completely induced trace norm one. These facts apply to arbitrary operators and remain valid with any ancilla. Therefore $\|\mathcal R\|_\diamond\le2$, and
\begin{equation}\label{eq:block-bound-NLS}
 e_j(p)\le2p^2+2\delta+2|f(p,s)|
             =(2+4g)p^2+2|f(p,s)|.
\end{equation}
There is no conversion from entrywise Choi norms and hence no input-dimension factor in this estimate.

The scalar estimate proved in Appendix~\ref{app:scalar} gives
$0\le f(p,p+2gp^2)\le4gp^{5/2}$ on this interval. Thus
\begin{equation}\label{eq:refined-NLS-bound}
 e_j(p)\le(2+4g+8g\sqrt p)p^2\le(2+8g)p^2.
\end{equation}
This proves the upper bound.

For the lower bound, measure the environment with the two outcomes $|0\rangle\langle0|$ and $I-|0\rangle\langle0|$. For every input state, the true complementary output has probabilities $(1-p,p)$, while the constructed output has $(1-s,s)$. Trace-norm contractivity under this measurement gives $e_j(p)\ge2(s-p)=4gp^2$.
\end{proof}

The coefficient $2+32/(d^2-1)$ in the analytical ceiling decreases
with system dimension, while the guaranteed interval $0\le p\le1/8$
is common to all spins. For a prescribed error tolerance $\tau$,
the bound gives the sufficient condition
\begin{equation}\label{eq:tolerance-window}
 p\le\min\left\{\frac18,\sqrt{\frac{3\tau}{38}}\right\}
\end{equation}
for $e_j(p)\le\tau$ at every spin. The common adjoint action
underlying this construction will be formulated as a general
sufficient condition in Sec.~\ref{sec:mechanism}.

\section{Capacity consequences}\label{sec:capacity}\label{sec:capacity-approx}
We now translate the reconstruction estimate into an explicit
capacity bracket, combining a coherent-information lower bound
with a conditional-entropy upper bound.

For a channel $\N$, let
$I_c(\N,\rho)=S(\N(\rho))-S(\N^c(\rho))$. The maximally mixed input $\pi=I/d$ gives an explicit achievable benchmark for NLS. Indeed, Eq.~\eqref{eq:spin-traces} implies
\begin{equation}\label{eq:mixed-outputs}
 \M_{j,p}(\pi)=\pi,\qquad
 \M_{j,p}^c(\pi)=\operatorname{diag}(1-p,p/3,p/3,p/3).
\end{equation}
Define
\begin{align}
 q_j(x)&=\log_2d-\hb(x)-x\log_2 3,\label{eq:q-benchmark}\\
 G(\epsilon)&=2\epsilon+
 \left(1+\frac\epsilon2\right)
 \hb\left(\frac{\epsilon}{2+\epsilon}\right),\label{eq:G-penalty}
\end{align}
where $\hb$ is binary entropy, extended continuously to its endpoints. Then $I_c(\M_{j,p},\pi)=q_j(p)$, so $Q(\M_{j,p})\ge\max\{0,q_j(p)\}$.

For the upper bound, the conditional-entropy formulation of
approximate degradability allows covariance to resolve the input
optimization. For a specified degrading map $\D$ with $\|\N^c-\D\N\|_\diamond\le\epsilon$, Sutter et al.~\cite{sutter2017approximate} give
\begin{equation}\label{eq:conditional-capacity-bound}
 Q(\N)\le U_{\D}(\N)+G(\epsilon),
\end{equation}
when the complementary output dimension is four, where
\begin{equation}\label{eq:UD-definition}
 U_{\D}(\N)=\max_\rho\{S(\N(\rho))-S(\D\N(\rho))\}.
\end{equation}
This is the specialization to a four-dimensional environment of
Theorem~3.4(ii) in the arXiv version of
Ref.~\cite{sutter2017approximate}; the diamond error here is
the full norm.

\begin{theorem}[Capacity bound from approximate degradability]\label{thm:capacity}
For every spin $j$, $0\le p\le1/8$, and $s=p+2p^2/\kappa$,
\begin{align}
 Q(\M_{j,p})&\ge\max\{0,q_j(p)\},\label{eq:capacity-lower}\\
 Q(\M_{j,p})&\le\min\{\log_2d,\ q_j(s)+G(\overline\varepsilon_j(p))\}.
 \label{eq:capacity-upper}
\end{align}
In particular,
\begin{equation}\label{eq:capacity-remainder}
 0\le Q(\M_{j,p})-q_j(p)
 \le G\left(\frac{38}{3}p^2\right).
\end{equation}
The remainder bound is uniform in the spin.
\end{theorem}
\begin{proof}
Take $\N=\M_{j,p}$ and $\D=\M_{j,s}^c$. If $W:B\to E'F$ is a Stinespring isometry for $\D$, then
\begin{equation}\label{eq:UD-conditional-entropy}
 S(\N(\rho))-S(\D\N(\rho))
 =S(F\mid E')_{W\N(\rho)W^\dagger}.
\end{equation}
Conditional entropy is concave, so this is concave in $\rho$. Both maps are covariant, with the complementary representation described after Eq.~\eqref{eq:NLS-covariance}. The objective is therefore invariant under rotating the input. Averaging over $SU(2)$ and using irreducibility sends any input to $\pi$; concavity shows that this average cannot decrease the objective. Thus its maximum is attained at $\pi$ and Eq.~\eqref{eq:mixed-outputs} gives
\begin{equation}\label{eq:UD-evaluation}
 U_{\D}(\N)=\log_2d-\hb(s)-s\log_2 3=q_j(s).
\end{equation}
This argument does not assume concavity of the ordinary coherent information or identify its exact optimizer. Applying Theorem~\ref{thm:main} in Eq.~\eqref{eq:conditional-capacity-bound} proves the two-sided estimate.

Finally, $p\le s\le1/6$ and
$q_j'(x)=\log_2[x/(3(1-x))]<0$ for $0<x<3/4$. Hence $q_j(s)\le q_j(p)$. The function $G$ is increasing on nonnegative arguments, so Eqs.~\eqref{eq:epsilon-NLS} and \eqref{eq:capacity-upper} imply Eq.~\eqref{eq:capacity-remainder}.
\end{proof}

The fixed four-dimensional environment makes the entropy-continuity penalty independent of the spin. The penalty may nevertheless be appreciable near $p=1/8$: the common validity interval need not give a narrow capacity bracket throughout that interval. Since $G(cp^2)=O(p^2\log(1/p))$ for fixed $c>0$, Eq.~\eqref{eq:capacity-remainder} implies
\begin{align}
 Q(\M_{j,p})={}&\log_2(2j+1)-\hb(p)-p\log_2 3\nonumber\\
 &+O\bigl(p^2\log(1/p)\bigr),\qquad p\downarrow0,
 \label{eq:capacity-expansion}
\end{align}
with a remainder uniform in the spin. This estimate controls any gain above the maximally mixed input's coherent information, including input optimization and regularization.

The expansion in Eq.~\eqref{eq:capacity-expansion} is also implied by the established degradable flagged-extension construction of Kianvash, Fanizza, and Giovannetti for mixtures of the identity with arbitrary noise \cite{kianvash2022bounding}; Roofeh and Karimipour evaluated its spin-one specialization \cite{roofeh2024noisy}. Theorem~\ref{thm:capacity} makes the capacity consequence of our explicit degrading map and its uniform error bound quantitative.

\section{Comparison with optimal degrading maps}\label{sec:numerics}

We assess the performance of the analytical construction by
comparing its error with the optimum over all degrading channels:
\begin{align}
 \eta_j(p)&=\inf_{\D\ {\rm CPTP}}
 \bigl\|\M^c_{j,p}-\D\circ\M_{j,p}\bigr\|_\diamond,
 \nonumber\\
 e_j(p)&=\bigl\|\M^c_{j,p}
 -\M^c_{j,s}\circ\M_{j,p}\bigr\|_\diamond,
 \quad s=p+\frac{2p^2}{j(j+1)}.
 \label{eq:numerical-errors}
\end{align}
Both use the full diamond norm, without a factor of one half. The sampled
spins are $j=1,3/2,2,5/2$, and the noise grid is
$p=0.01,0.02,0.04,0.06,0.08,0.10$.

The optimization ranges over all CPTP degrading maps, without a covariance
or analytical-ansatz constraint. Only the lower-witness search fixes an
auxiliary dual density operator to $I_d/d$; this leaves the primal
degrading-map optimization unrestricted. Figure~\ref{fig:optimal-comparison} shows
that the construction is close to the numerical optimum: the repaired
numerical brackets place $e_j(p)/\eta_j(p)-1$ below $2.48\%$ at every
sampled point.

Let $[\ell_j(p),u_j(p)]$ denote the repaired numerical bracket for
$\eta_j(p)$. All brackets are positive, their relative width
$(u_j(p)-\ell_j(p))/\ell_j(p)$ is below $0.235\%$, and their largest
absolute width is below $9.62\times10^{-7}$.
For each sampled spin, the largest upper estimate of $\eta_j(p)/p^2$
exceeds its smallest lower estimate across the grid by less than
$3.34\%$. Thus the normalized error varies little on this window,
supporting near-quadratic behavior there.

These brackets quantify performance on the sampled grid using
repaired floating-point witnesses. Solver statuses and feasibility
diagnostics are reported in Appendix~\ref{app:numerical-validation};
the brackets are not interval-arithmetic certificates. The asymptotic
order of $\eta_j(p)$ and asymptotic optimality of the construction
remain open.

\begin{figure*}[t]
 \centering
 \includegraphics[width=\textwidth]{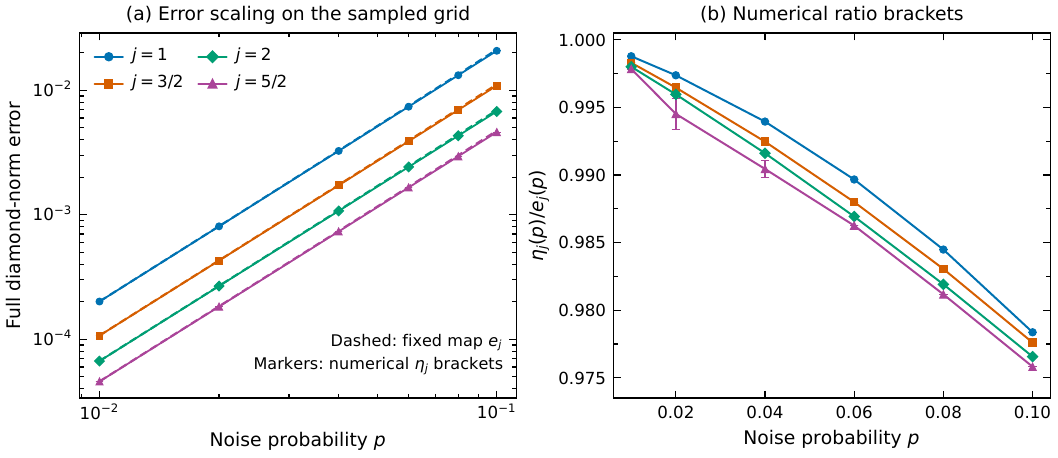}
    \caption{Finite-grid comparison of the analytical degrading map
    and numerical brackets for the optimum over all CPTP degrading maps.
    (a) Full diamond-norm errors for $j=1,3/2,2,5/2$.
    Dashed curves show numerical evaluations of $e_j(p)$ for
    $s=p+2p^2/[j(j+1)]$; markers represent the numerical brackets
    for $\eta_j(p)$.
    (b) Propagated numerical brackets for $\eta_j(p)/e_j(p)$;
    values closer to one indicate that the analytical construction
    is closer to the optimum.
    In both panels, markers lie at bracket midpoints and vertical
    error bars indicate the lower and upper endpoints.
    Error bars are plotted at every sampled point; narrow brackets
    are hidden by the marker symbols.
    These intervals describe numerical bounds, not statistical
    confidence intervals.
    Marker shapes and colors distinguish the spins; connecting
    lines are visual guides through
    $p=0.01,0.02,0.04,0.06,0.08,0.10$.
    The brackets use repaired floating-point witnesses and
    feasibility checks, not interval-arithmetic certificates.
    The finite-grid data do not establish an asymptotic lower
    bound on $\eta_j(p)$.}
 \label{fig:optimal-comparison}
\end{figure*}

\section{An adjoint-eigenoperator criterion and further families}\label{sec:mechanism}
The NLS construction is driven by a common adjoint eigenvalue on
the noise operators. Building on the Pauli cancellation strategy of
Ref.~\cite{leditzky2018prl}, we isolate this relation as a sufficient
condition for a single parameter shift to yield a quadratic
reconstruction-error bound.

\begin{proposition}[Common adjoint-eigenvalue criterion]\label{prop:rigorous_condition}
Let $\LL(X)=\sum_{k=1}^K L_kXL_k^\dagger$ be a fixed CPTP map and $\N_p=(1-p)\id+p\LL$ for $0\le p\le1$. Suppose that for one real $\gamma$,
\begin{equation}\label{eq:adjoint-criterion}
 \LL^\dagger(L_k)=(1-\gamma)L_k\qquad(k=1,\ldots,K).
\end{equation}
Then $0\le\gamma\le2$. Use complements associated with
$\{\sqrt{1-p}I,\sqrt pL_1,\ldots,\sqrt pL_K\}$. For $0\le p\le1/8$ and $s=p+2\gamma p^2$, one has $s\le3/16$ and
\begin{equation}\label{eq:general-uniform}
 \|\N_p^c-\N_s^c\circ\N_p\|_\diamond
 \le(2+8\gamma)p^2\le18p^2.
\end{equation}
Both the interval and the last constant are independent of system dimension and $K$.
\end{proposition}
\begin{proof}
The unital completely positive map $\LL^\dagger$ is an operator-norm contraction. Applying this to a nonzero $L_k$ gives $|1-\gamma|\le1$. Equation~\eqref{eq:adjoint-criterion} implies the exact affine relation
$\N_p^\dagger(L_k)=(1-\gamma p)L_k$.

The block estimate in the proof of Theorem~\ref{thm:main} extends directly to arbitrary, possibly non-Hermitian, Kraus operators. Specifically, the scalar block has norm $s-p$, and
\begin{equation}
 p\LL^c-s\LL^c\N_p
 =p\LL^c(\id-\N_p)-(s-p)\LL^c\N_p
\end{equation}
has norm at most $2p^2+(s-p)$. In Eq.~\eqref{eq:cross-isometries}, replace the sum over three Hermitian operators by the sum over the $K$ Kraus operators. The isometry argument still gives cross-map norm at most two. Preservation of adjoints by $\N_p^\dagger$ gives the same real coefficient in both cross blocks. Consequently,
\begin{align}
 \|\N_p^c-\N_s^c\N_p\|_\diamond
 &\le2p^2+2(s-p)+2|f_\gamma(p,s)|,\label{eq:general-block-bound}\\
 f_\gamma(p,s)&=\sqrt{p(1-p)}-(1-\gamma p)\sqrt{s(1-s)}.
 \label{eq:general-cross}
\end{align}
Appendix~\ref{app:scalar} proves $0\le f_\gamma(p,p+2\gamma p^2)\le2\gamma p^2$ on the stated interval. Inserting $s-p=2\gamma p^2$ proves the bound.
\end{proof}

The proof uses the adjoint action on the noise Kraus operators
and requires neither tracelessness nor a forward eigen-operator
condition. If the $L_k$ are also traceless,
Theorem~\ref{thm:diamond_distance} additionally gives the calibration
$\|\N_p-\id\|_\diamond=2p$.

The common eigenvalue allows all noise operators to be treated
by a single parameter shift. For unequal Pauli noise weights,
Ref.~\cite{leditzky2018prl} permits distinct corrections for different
errors. Covariance alone does not ensure a common adjoint eigenvalue.

The following examples illustrate three roles of the criterion:
recovering the familiar depolarizing construction, exploiting Casimir
identities in Lie-generator channels, and controlling a
Werner--Holevo family whose number of Kraus operators grows with
dimension.

\subsection{Uniform Weyl weights: the depolarizing example}\label{sec:depolarizing}
For $d\ge2$, let $W_0=I,W_1,\ldots,W_{d^2-1}$ be an orthogonal Weyl unitary basis. The uniform nonidentity noise map is
\begin{equation}\label{eq:uniform-Weyl}
 \LL_{\rm dep}(X)=\frac1{d^2-1}\sum_{a=1}^{d^2-1}W_aXW_a^\dagger
 =\frac{d\,\tr(X)I-X}{d^2-1}.
\end{equation}
The normalized Kraus operators of this depolarizing channel are
$L_a=W_a/\sqrt{d^2-1}$. Each is traceless and has adjoint noise eigenvalue $-1/(d^2-1)$, so
\begin{equation}\label{eq:depolarizing-gamma}
 \gamma=\frac{d^2}{d^2-1},\qquad
 s=p+\frac{2d^2}{d^2-1}p^2.
\end{equation}
The qubit value is $s=p+(8/3)p^2$, agreeing with the known Pauli cancellation. Proposition~\ref{prop:rigorous_condition} therefore supplies a
diamond-error bound uniform in dimension for this family.

\subsection{Irreducible \texorpdfstring{$SU(N)$}{SU(N)} Lie-generator channels}\label{sec:Lie-families}
Let $R$ be a nontrivial finite-dimensional irreducible unitary representation of $SU(N)$ with Hermitian generators $T_a$ in a common invariant normalization. Let $C_R$ and $C_A$ be its representation and adjoint quadratic Casimirs. Define the established Lie-generator channel family \cite{ritter2005quantum}
\begin{align}
 \LL_R(X)&=\frac1{C_R}\sum_{a=1}^{N^2-1}T_aXT_a,
 &\sum_aT_a^2&=C_RI,\label{eq:Lie-channel}\\
 \N_{R,p}&=(1-p)\id+p\LL_R.&&\nonumber
\end{align}
The Casimir relation makes the noise map CPTP, unital, and self-adjoint. The generators transform by the orthogonal adjoint representation, so the same Kraus-sum argument used for Eq.~\eqref{eq:NLS-covariance} proves covariance. Related $U(d)$ channel constructions, including noisy mixtures,
are discussed in Sec.~7 of Ref.~\cite{karimipour2024noisy}.

The adjoint Casimir identity gives
\begin{align}
 \sum_a[T_a,[T_a,T_b]]&=C_AT_b,\nonumber\\
 \sum_aT_aT_bT_a&=\left(C_R-\frac{C_A}{2}\right)T_b.
 \label{eq:Casimir-triple}
\end{align}
For $L_b=T_b/\sqrt{C_R}$, this verifies Eq.~\eqref{eq:adjoint-criterion} with $\gamma_R=C_A/(2C_R)$.
\begin{corollary}[Lie-generator error bound]\label{cor:Lie-bound}
For every such representation and $0\le p\le1/8$, set $s=p+(C_A/C_R)p^2$. Then
\begin{equation}\label{eq:Lie-uniform-bound}
 \|\N_{R,p}^c-\N_{R,s}^c\circ\N_{R,p}\|_\diamond
 \le\left(2+\frac{4C_A}{C_R}\right)p^2\le18p^2.
\end{equation}
\end{corollary}
The interval and last coefficient are independent of the representation and of $N$, since Proposition~\ref{prop:rigorous_condition} ensures $0\le\gamma_R\le2$. For $SU(2)$, the usual normalization gives $C_A=2$, $C_R=j(j+1)$, and recovers the NLS construction. The defining representation of $SU(N)$ gives a depolarizing
channel. Higher representations also provide examples outside the
mixed-unitary class, as the following $SU(3)$ case demonstrates.

For the adjoint representation of $SU(3)$, with fundamental normalization $\tr(t_at_b)=\delta_{ab}/2$, one has
\begin{equation}\label{eq:SU3-adjoint}
 d_R=8,\qquad C_R=C_A=3,\qquad
 \gamma_R=\frac12,\qquad s=p+p^2.
\end{equation}
For $0<p<1$, $I,T_1,\ldots,T_8$ are linearly independent, so the channel has Choi rank nine. At the same input dimension and for $0<p<1$, the spin-$7/2$ NLS channel has rank four and the uniform depolarizing channel has rank 64. More strongly, the $SU(3)$ adjoint example is not mixed-unitary, as proved in Appendix~\ref{app:adjoint}. It therefore lies outside random-Weyl channels, independently of these rank comparisons.

\subsection{Noisy Werner--Holevo channels}\label{sec:WH-family}
Consider the established $SO(d)$-covariant family \cite{roofeh2024noisy,karimipour2024noisy}
\begin{equation}\label{eq:WH-channel}
 \Phi_p(X)=(1-p)X+\frac p{d-1}\bigl(\tr(X)I-X^{\mathsf T}\bigr).
\end{equation}
For $m<n$, set
\begin{equation}
 L_{mn}=\frac{-i}{\sqrt{d-1}}
       (|m\rangle\langle n|-|n\rangle\langle m|).
\end{equation}
These Hermitian, traceless Kraus operators satisfy
\begin{align}
 \sum_{m<n}L_{mn}^2&=I,\nonumber\\
 \sum_{m<n}L_{mn}XL_{mn}&=
             \frac{\tr(X)I-X^{\mathsf T}}{d-1}.
 \label{eq:WH-Kraus}
\end{align}
Because $L_{mn}^{\mathsf T}=-L_{mn}$, the self-adjoint noise map has eigenvalue $1/(d-1)$ on every $L_{mn}$. Thus
\begin{equation}\label{eq:WH-gamma}
 \gamma_d=\frac{d-2}{d-1},\qquad
 s=p+\frac{2(d-2)}{d-1}p^2.
\end{equation}
For every $d\ge3$, Proposition~\ref{prop:rigorous_condition} gives
\begin{equation}\label{eq:WH-bound}
 \|\Phi_p^c-\Phi_s^c\circ\Phi_p\|_\diamond
 \le\left(2+8\frac{d-2}{d-1}\right)p^2\le10p^2
\end{equation}
on $0\le p\le1/8$. The $d=3$ case recovers spin one. For $d\ge4$ and $0<p<1$,
the Choi rank $1+d(d-1)/2$ distinguishes this family from NLS.

Capacity bounds for related two-parameter Werner--Holevo families have also been studied \cite{roofeh2025capacities}. The result here is a uniform reconstruction-error bound for
Eq.~\eqref{eq:WH-channel}. For both the Lie-generator and
Werner--Holevo families, growing environment dimensions must also
be accounted for in capacity bounds: a uniform diamond-error estimate
alone does not give a dimension-independent entropy-continuity
penalty.

\section{Conclusion}\label{sec:conclusion}

We have constructed degrading maps for the finite-spin NLS
family with explicit quadratic error bounds on a noise interval
common to all spins. The common adjoint action on the three
spin generators allows a single noise-parameter shift to cancel
the leading reconstruction error, giving the bound
$[2+8/(j(j+1))]p^2$. A lower bound of the same order establishes
quadratic scaling for this construction at each fixed spin.
Floating-point SDP comparisons on the reported finite grid
place its error within $2.48\%$ of the numerical optimum.

This reconstruction guarantee also yields a directly computable
quantum-capacity upper bound: covariance and conditional-entropy
concavity resolve the conditional-entropy optimization.
Combining this upper bound with the coherent information of
the maximally mixed input gives explicit two-sided capacity
bounds with a spin-uniform remainder.

The quadratic constructions of Leditzky, Leung,
and Smith concern Pauli channels within the mixed-unitary
class \cite{leditzky2018prl}. The approximate-degradability
analysis of generalized amplitude damping by Khatri, Sharma,
and Wilde concerns qubit channels
\cite{khatri2020generalized}.
Building on the low-noise Pauli cancellation strategy, we
provide a directly checkable sufficient condition: the adjoint
noise channel maps each of its fixed Kraus operators to the
same scalar multiple of itself. This common action allows
a single noise-parameter shift to cancel the leading
reconstruction-error contributions.
The resulting quadratic bounds have explicit
constants and a common validity interval independent of
system dimension and Kraus count.
Casimir identities and the Werner--Holevo Kraus representation
verify the condition for established channel families.
In particular, the $SU(3)$ adjoint family
provides a unital example that remains outside the mixed-unitary
class for $0<p<1$, while admitting a reconstruction-error bound
of $6p^2$ on $0\le p\le1/8$.

Future work can pursue both sharper reconstruction guarantees
and a broader range of noise models. For NLS channels, this
includes determining the optimal low-noise scaling of the
degradability error, particularly whether a positive quadratic
lower bound holds, and improving the constants and validity
intervals of the constructive bounds. For more general channels,
a natural next step is to extend the degrading-map construction
to noise operators with distinct adjoint eigenvalues, thereby
relaxing the common-scalar condition. Together, these directions
would help clarify when channel structure supports quadratic
reconstruction guarantees that remain uniform as the system
dimension grows.

\begin{acknowledgments}
The authors used ChatGPT (OpenAI; GPT-6 Astra) as an assistive tool for language refinement, manuscript organization, literature searches, and the development and debugging of semidefinite programming scripts. All mathematical derivations, numerical methodologies, reported results, scientific interpretations, and conclusions were independently verified and approved by the authors.
\end{acknowledgments}

\appendix
\section{Scalar estimates for the cross-block coefficient}\label{app:scalar}
Let $0\le\gamma\le2$, $0\le p\le1/8$, and $s=p+2\gamma p^2$. Write
\begin{equation}
 A=\sqrt{p(1-p)},\qquad B=(1-\gamma p)\sqrt{s(1-s)}.
\end{equation}
For $p>0$, $s\le3/16$, $A^2\ge7p/8$, and $B^2\ge117p/256$, so $A+B\ge(3/2)\sqrt p$. Direct expansion gives
\begin{align}
 A^2-B^2&=\gamma p^3H_\gamma(p),\label{eq:rationalization}\\
 H_\gamma(p)&=2+3\gamma-(3\gamma+2\gamma^2)p\nonumber\\
 &\hspace{1em}-4\gamma^2p^2+4\gamma^3p^3.
\end{align}
The last three terms combine into
$-\gamma p[3+2\gamma+4\gamma p(1-\gamma p)]\le0$. Also
\begin{equation}
 H_\gamma(p)\ge2+\frac{21}{8}\gamma-\frac5{16}\gamma^2\ge2.
\end{equation}
Thus $0\le H_\gamma(p)\le2+3\gamma\le8$, and rationalization gives
\begin{align}
 0\le A-B&=\frac{\gamma p^3H_\gamma(p)}{A+B}
 \le\frac{16}{3}\gamma p^{5/2}\nonumber\\
 &\le\frac{4\sqrt2}{3}\gamma p^2\le2\gamma p^2.
 \label{eq:scalar-general-bound}
\end{align}
This proves the scalar estimate used in Proposition~\ref{prop:rigorous_condition}. The cases $p=0$ and $\gamma=0$ follow directly.

For NLS, $\gamma=g=1/\kappa\le4/3$, so $s\le1/6$, $gp\le1/6$, and $H_g(p)\le2+3g\le6$. The same denominator bound therefore yields
\begin{equation}\label{eq:scalar-NLS-bound}
 0\le f(p,p+2gp^2)\le4gp^{5/2}.
\end{equation}
Together with Eq.~\eqref{eq:block-bound-NLS}, this gives the refined estimate in Eq.~\eqref{eq:refined-NLS-bound} and the stated spin-uniform coefficient.

\section{The \texorpdfstring{$SU(3)$}{SU(3)} adjoint example is not mixed-unitary}\label{app:adjoint}
Realize the adjoint representation on the eight-dimensional Hilbert space of traceless complex $3\times3$ matrices with the Hilbert--Schmidt inner product. Its generators are $T_a=\operatorname{ad}(t_a)$, where $\operatorname{ad}(X)(Y)=[X,Y]$ and the $t_a$ are Hermitian fundamental generators. For $0<p<1$, the Kraus span of the noisy channel is
\begin{equation}\label{eq:adjoint-Kraus-span}
 \mathcal S=\{\alpha I+\operatorname{ad}(X):
                    \alpha\in\mathbb C,\ \tr X=0\}.
\end{equation}
Suppose $U=\alpha I+\operatorname{ad}(X)$ is unitary. Since $\operatorname{ad}(X)^\dagger=\operatorname{ad}(X^\dagger)$, normality implies
\begin{equation}
 0=[U,U^\dagger]=\operatorname{ad}([X,X^\dagger]).
\end{equation}
A matrix commuting with every traceless matrix is scalar. Its commutator trace is zero, so $[X,X^\dagger]=0$ and $X$ is normal. Diagonalize it with eigenvalues $z_1,z_2,z_3$. On the two-dimensional traceless diagonal subspace, $U$ has eigenvalue $\alpha$; on off-diagonal matrix units, its eigenvalues are $\alpha+z_i-z_j$. Unitarity therefore requires
\begin{equation}
 |\alpha|=|\alpha+(z_i-z_j)|=|\alpha-(z_i-z_j)|=1.
\end{equation}
Adding the squares of the last two moduli and using $|\alpha|=1$ forces $z_i-z_j=0$ for every pair. Tracelessness then gives $X=0$. Thus $\mathcal S$ contains only scalar unitaries.

Every Kraus operator in any representation of the same channel belongs to its Choi support, equivalently to $\mathcal S$ after vectorization. A mixed-unitary representation would therefore use only scalar unitaries and produce the identity channel. This contradicts Choi rank nine for $0<p<1$. The example is consequently not mixed-unitary and in particular is not a random-Weyl channel.

\section{Numerical formulation and verification}
\label{app:numerical-validation}

This appendix presents the SDP formulations, computational
implementation, and numerical verification procedures.
ChatGPT (OpenAI; GPT-6 Astra) assisted with script development
and debugging. The authors specified the formulations and
methodology, reviewed and tested the code, and independently
verified all reported numerical results.

\subsection{Optimization over degrading maps}
Throughout this appendix, we use the unnormalized, input-first
Choi convention in Eq.~\eqref{eq:choi-convention}.
Write $d=2j+1$, $d_E=4$, $C=J(\M^c_{j,p})$, and
$X=J(\D)$. If $K_\mu$ are the Kraus operators of
$\M_{j,p}$, input-side precomposition is
$\mathcal P_p(X)=\sum_\mu A_\mu X A_\mu^\dagger$, where
$A_\mu=K_\mu^T\otimes I_E$.
The unrestricted CPTP optimization is the semidefinite program
\begin{equation}\label{eq:optimal-sdp}
 \begin{aligned}
  \text{minimize}\quad &t\\
  \text{subject to}\quad &X\succeq0,\qquad
              \tr_E X=I_d,\\
  &Z\succeq C-\mathcal P_p(X),\quad
   Z\succeq-C+\mathcal P_p(X),\\
  &\tr_E Z\preceq tI_d.
 \end{aligned}
\end{equation}
Here $X$ and $Z$ are Hermitian. No covariance or analytical-ansatz
constraint is imposed on $\D$. In the implementation, the
discrepancy in the $Z$ constraints is divided by $p^2$ and the resulting
objective is multiplied by $p^2$.

A complementary lower-witness search uses the dual constraints
\begin{equation}\label{eq:dual-witness}
 \begin{aligned}
  \text{maximize}\quad &\tr(CW)-\tr H\\
  \text{subject to}\quad &\rho\succeq0,\quad\tr\rho=1,\\
  &-\rho\otimes I_E\preceq W\preceq\rho\otimes I_E,\\
  &H\otimes I_E\succeq\mathcal P_p^*(W).
 \end{aligned}
\end{equation}
We fix $\rho=I_d/d$ in this search. This restriction supplies valid lower
witnesses and does not restrict the primal degrading map.

\subsection{Channel construction and numerical witnesses}
The spin basis is ordered $m=j,j-1,\ldots,-j$, and the environment basis
is $(0,x,y,z)$. With $K_0=\sqrt{1-p}\,I_d$ and
$K_k=\sqrt{p/[j(j+1)]}\,J_k$, the complementary channel is implemented as
$[\M^c_{j,p}(\sigma)]_{\mu\nu}
=\tr(K_\mu\sigma K_\nu^\dagger)$.
The Choi matrices are formed using
Eq.~\eqref{eq:choi-convention}, with the environment as the output
factor.
The precomposition map in Eq.~\eqref{eq:optimal-sdp} is checked against
independent propagation of every input matrix unit. Its adjoint is
$\mathcal P_p^*(W)=\sum_\mu A_\mu^\dagger W A_\mu$, which is checked
using the Hilbert--Schmidt inner product.

The primal repair first hermitizes $X$, clips its eigenvalues below
$10^{-13}$, and writes the resulting matrix as $X_+$.
With $S=\tr_E X_+$, the repaired channel has Choi matrix
\begin{equation}
 \widehat X=(S^{-1/2}\otimes I_E)X_+
                  (S^{-1/2}\otimes I_E).
\end{equation}
This construction is positive and trace preserving in exact arithmetic.
For its recomputed discrepancy $F=C-\mathcal P_p(\widehat X)$,
$|F|\succeq\pm F$ supplies the upper witness
$u=\lambda_{\max}(\tr_E|F|)$.
For the dual repair, $W$ is hermitized and, when necessary, rescaled until
$\|W\|_{\mathrm{op}}\le(1-10^{-12})/d$.
The Hermitian matrix $H$ is shifted by
\begin{equation}
 \left[\max\!\left\{0,
 -\lambda_{\min}(H\otimes I_E-\mathcal P_p^*(W))\right\}
 +10^{-12}\right]I_d.
\end{equation}
Its objective in Eq.~\eqref{eq:dual-witness}, floored at zero, gives the
repaired numerical lower witness $\ell$.

For the fixed analytical map, let
$F_{\mathrm{fix}}=C-\mathcal P_p(J(\M^c_{j,s}))$.
Independent spectral witnesses give
\begin{equation}
 \frac{\|F_{\mathrm{fix}}\|_1}{d}\le e_j(p)
 \le\lambda_{\max}\!\left(\tr_E|F_{\mathrm{fix}}|\right).
\end{equation}
These bounds agree within $1.15\times10^{-16}$ over the grid.
Auxiliary fixed-map diamond-norm SDPs agree with the spectral upper
witness within $3.97\times10^{-11}$.

\subsection{Solver settings and verification}

The SDPs were implemented in CVXPY~1.9.3 using Clarabel~0.11.1
and SCS~3.3.1. Clarabel computed the complete 24-point baseline
with absolute-gap, relative-gap, and feasibility tolerances of
$10^{-10}$ and a limit of 250 iterations. All baseline solves
returned \texttt{optimal\_inaccurate}. Supplementary SCS
lower-witness searches completed 19 points, using absolute and
relative tolerances of $10^{-9}$ and a limit of $2\times10^5$
iterations. At each point, the final lower witness is the larger
repaired value from the available searches; all upper witnesses
come from the Clarabel baseline.

Numerical accuracy was assessed through the repaired witnesses
and their feasibility diagnostics. The maximum trace-preservation
residual was below $3.2\times10^{-15}$, and the checked positivity
and dual-feasibility slacks were nonnegative. Known-norm tests
and the analytical error bounds supplied additional consistency
checks. All 24 final brackets satisfy the reporting criterion
$\ell>0$ and $(u-\ell)/\ell<1\%$; the largest relative width is
below $0.235\%$. Requested solver tolerances are not interpreted
as achieved errors.

The normalized-error variation reported in
Sec.~\ref{sec:numerics} is
\[
 \frac{\max_{p\in\mathcal P}\{u_j(p)/p^2\}}
      {\min_{p\in\mathcal P}\{\ell_j(p)/p^2\}}-1,
\]
evaluated separately for each spin on the sampled grid
$\mathcal P$. Its maximum over the four spins is below $3.34\%$.
These diagnostics describe floating-point feasibility and
resolution, not interval-arithmetic or rational certification.
Complete solver logs, witness-selection records, and pointwise
diagnostics are retained in the computational package.

% BEGIN INLINE SUPPLEMENTARY APPENDICES
% Supplementary pedagogical derivations adapted from original Appendices B, C, E.
% Current notation and complementary-channel conventions are used throughout.

\section{Supplementary componentwise Taylor calculation}
\label{app:supplementary-taylor}

Appendices~\ref{app:supplementary-taylor}--\ref{app:supplementary-covariance}
provide supplementary, pedagogical details of the Taylor expansion,
singlet calculation, and covariance identities. They introduce no additional
assumptions or results needed for the main theorems. The uniform error
bound is proved in Sec.~\ref{sec:main_result} and Appendix~\ref{app:scalar},
and the capacity bound is proved in Sec.~\ref{sec:capacity}. The calculations
below make the componentwise cancellation explicit; their fixed-spin
Taylor remainders do not replace the uniform estimates used in those proofs.

\subsection{Complementary-channel entries}

Write $\kappa=j(j+1)$ and define
\begin{equation}\label{eq:supp-cj}
 p_j=\frac{p}{\kappa},\qquad
 c_j(p)=\sqrt{\frac{p(1-p)}{\kappa}}.
\end{equation}
For an arbitrary input operator $X$, let $B=\M_{j,p}^c(X)$.
In the environment basis $(0,x,y,z)$, Eq.~\eqref{eq:complement-convention}
and the Hermiticity of the spin generators give, for $k,l=1,2,3$,
\begin{subequations}\label{eq:supp-complement-entries}
\begin{align}
 B_{00}&=(1-p)\tr X,\\
 B_{0k}&=c_j(p)\tr(XJ_k)
        =c_j(p)\tr(J_kX),\\
 B_{k0}&=c_j(p)\tr(J_kX),\\
 B_{kl}&=p_j\tr(J_kXJ_l).
\end{align}
\end{subequations}
The two cross entries agree by cyclicity of the trace.
For Hermitian $X$, the output is Hermitian and
$B_{k0}=B_{0k}^*$ as well. The latter conjugation relation
is not imposed on a general input matrix unit.

\subsection{Moments after the noisy channel}

For $s=p+ap^2$, consider
$\Psi_{p,s}=\M_{j,s}^c\circ\M_{j,p}$ and put
$X'=\M_{j,p}(X)$. Trace preservation gives $\tr X'=\tr X$.
The first spin moments satisfy
\begin{align}
 \tr(J_kX')
 &=(1-p)\tr(J_kX)
   +\frac{p}{\kappa}\sum_l\tr(J_kJ_lXJ_l)
   \nonumber\\
 &=(1-p)\tr(J_kX)
   +\frac{p}{\kappa}\tr\!\left[
       \left(\sum_lJ_lJ_kJ_l\right)X\right]
   \nonumber\\
 &=\left(1-\frac{p}{\kappa}\right)\tr(J_kX).
 \label{eq:supp-first-moment}
\end{align}
The last equality uses
$\sum_lJ_lJ_kJ_l=(\kappa-1)J_k$, established in
Eq.~\eqref{eq:spin-triple}.
The complementary output of $X'$ therefore has entries
\begin{subequations}\label{eq:supp-composite-entries}
\begin{align}
 [\Psi_{p,s}(X)]_{00}&=(1-s)\tr X,\\
 [\Psi_{p,s}(X)]_{0k}
 &=[\Psi_{p,s}(X)]_{k0}\nonumber\\
 &=c_j(s)\left(1-\frac{p}{\kappa}\right)\tr(J_kX),\\
 [\Psi_{p,s}(X)]_{kl}&=\frac{s}{\kappa}\tr(J_kX'J_l).
\end{align}
\end{subequations}
Thus all three identity--noise cross terms have the same scalar
coefficient.

\subsection{The Taylor cancellation}

Set $\Phi_{p,s}=\M_{j,p}^c-\Psi_{p,s}$ and $g=1/\kappa$.
Its cross entries are
\begin{align}
 [\Phi_{p,s}(X)]_{0k}
 &=[\Phi_{p,s}(X)]_{k0}
   =\Delta_c(p,s)\tr(J_kX),\nonumber\\
 \Delta_c(p,s)&=c_j(p)-(1-gp)c_j(s).
 \label{eq:supp-cross-coefficient}
\end{align}
All expansions in this subsection are as $p\downarrow0$ at
fixed $j$ and fixed $a$. First,
\begin{equation}\label{eq:supp-cp-expansion}
 c_j(p)=\frac{\sqrt p}{\sqrt\kappa}
 \left(1-\frac p2-\frac{p^2}{8}+O(p^3)\right).
\end{equation}
For the shifted parameter, one has
\begin{align}
 \sqrt{s}
 &=\sqrt p\left(1+\frac a2p-\frac{a^2}{8}p^2
                    +O(p^3)\right),\nonumber\\
 \sqrt{1-s}
 &=1-\frac p2-\left(\frac a2+\frac18\right)p^2
                    +O(p^3).
 \label{eq:supp-root-expansions}
\end{align}
Multiplication yields
\begin{equation}\label{eq:supp-cs-expansion}
 c_j(s)=\frac{\sqrt p}{\sqrt\kappa}
 \left[1+\frac{a-1}{2}p
       -\frac{a^2+6a+1}{8}p^2+O(p^3)\right].
\end{equation}
Including the attenuation of the first moment gives
\begin{align}
 &(1-gp)c_j(s)\nonumber\\
 &\quad=\frac{\sqrt p}{\sqrt\kappa}
 \left[1+\left(\frac{a-1}{2}-g\right)p\right.\nonumber\\
 &\qquad\left.
 -\left(\frac{a^2+6a+1}{8}+\frac{g(a-1)}2\right)p^2
 +O(p^3)\right].
 \label{eq:supp-composite-expansion}
\end{align}
Subtracting Eq.~\eqref{eq:supp-composite-expansion} from
Eq.~\eqref{eq:supp-cp-expansion}, we obtain
\begin{align}
 &\sqrt\kappa\,\Delta_c(p,p+ap^2)\nonumber\\
 &\quad=\left(g-\frac a2\right)p^{3/2}\nonumber\\
 &\qquad+\left(\frac{a^2+6a}{8}+\frac{g(a-1)}2\right)p^{5/2}
 +O(p^{7/2}).
 \label{eq:supp-delta-expansion}
\end{align}
Consequently the choice $a=2g=2/\kappa$ cancels the
$p^{3/2}$ term, leaving
\begin{equation}\label{eq:supp-cancelled-cross}
 \Delta_c(p,p+2gp^2)
 =\frac{g+\tfrac32g^2}{\sqrt\kappa}\,p^{5/2}
   +O(p^{7/2}).
\end{equation}
The coefficient $f(p,s)$ used in Sec.~\ref{sec:main_result}
satisfies $f(p,s)=\sqrt\kappa\,\Delta_c(p,s)$ because that
section uses the normalized operators $L_k=J_k/\sqrt\kappa$.

\subsection{The remaining entries}

The scalar discrepancy is exactly
\begin{equation}
 [\Phi_{p,s}(X)]_{00}=(s-p)\tr X=ap^2\tr X.
\end{equation}
For the noise block, the identity
$X-X'=p(X-\LL_j(X))$ gives
\begin{align}
 [\Phi_{p,s}(X)]_{kl}
 &=\frac{p-s}{\kappa}\tr(J_kXJ_l)\nonumber\\
 &\quad+\frac{sp}{\kappa}
       \tr\!\left[J_k\bigl(X-\LL_j(X)\bigr)J_l\right].
 \label{eq:supp-noise-entries}
\end{align}
Both terms are $O(p^2)$ at fixed spin when $s=p+ap^2$.
These identities explain the orders of the individual blocks.
The diamond-norm estimate, including arbitrary ancillas and a noise
interval common to all spins, is supplied by the block proof in
Sec.~\ref{sec:main_result} and the exact scalar estimates in
Appendix~\ref{app:scalar}.

\section{Supplementary singlet calculation}
\label{app:supplementary-singlet}

This appendix gives an explicit spin-basis illustration of the
vanishing local moments behind the orthogonality in
Theorem~\ref{thm:diamond_distance}. The theorem already has a
complete proof in Sec.~\ref{sec:geometry}; the singlet calculation
is not needed for the degrading-map or capacity bounds.

With $d=2j+1$, the normalized spin-$j$ singlet is
\begin{equation}\label{eq:supp-singlet}
 |\Psi_0\rangle=\frac{1}{\sqrt d}
 \sum_{m=-j}^{j}(-1)^{j-m}|j,m\rangle\otimes|j,-m\rangle.
\end{equation}
It spans the total-spin-zero subspace of $j\otimes j$.
We verify directly that
\begin{equation}\label{eq:supp-singlet-moments}
 \langle\Psi_0|(J_k\otimes I)|\Psi_0\rangle=0,
 \qquad k=x,y,z.
\end{equation}

\subsection{The longitudinal component}

Using $J_z|j,m\rangle=m|j,m\rangle$, we find
\begin{align}
 &\langle\Psi_0|(J_z\otimes I)|\Psi_0\rangle\nonumber\\
 &=\frac1d\sum_{m,m'=-j}^{j}
   (-1)^{(j-m)+(j-m')}m\nonumber\\
 &\qquad\times
   \langle j,m'|j,m\rangle\langle j,-m'|j,-m\rangle
   \nonumber\\
 &=\frac1d\sum_{m=-j}^{j}m=0.
\end{align}
The two inner products enforce $m'=m$, so the phase is one;
the last sum vanishes by pairing $m$ with $-m$.

\subsection{The transverse components}

Recall the ladder action in Eq.~\eqref{eq:ladder}, and write
$C_{j,m}^+=\sqrt{(j-m)(j+m+1)}$.
For the raising operator,
\begin{align}
 &\langle\Psi_0|(J_+\otimes I)|\Psi_0\rangle\nonumber\\
 &=\frac1d\sum_{m=-j}^{j-1}\sum_{m'=-j}^{j}
   (-1)^{2j-m-m'}C_{j,m}^+\nonumber\\
 &\qquad\times
   \underbrace{\langle j,m'|j,m+1\rangle}_{\delta_{m',m+1}}
   \underbrace{\langle j,-m'|j,-m\rangle}_{\delta_{m',m}}
   =0.
\end{align}
Every term vanishes because the two Kronecker constraints
cannot hold simultaneously. Hermitian conjugation gives
\begin{equation}
 \langle\Psi_0|(J_-\otimes I)|\Psi_0\rangle=0.
\end{equation}
Since $J_x=(J_++J_-)/2$ and $J_y=(J_+-J_-)/(2i)$,
both transverse expectations vanish, proving
Eq.~\eqref{eq:supp-singlet-moments}.

\subsection{Rotational invariance and exchange symmetry}

The singlet obeys
\begin{equation}
 (J_k\otimes I+I\otimes J_k)|\Psi_0\rangle=0,
 \qquad k=x,y,z.
\end{equation}
Thus the two local expectations have opposite signs.
On the other hand, the swap operator $P_{12}$ satisfies
\begin{equation}
 P_{12}|\Psi_0\rangle=(-1)^{2j}|\Psi_0\rangle.
\end{equation}
The exchange phase cancels in expectation values, so
\begin{equation}
 \langle\Psi_0|(J_k\otimes I)|\Psi_0\rangle
 =\langle\Psi_0|(I\otimes J_k)|\Psi_0\rangle.
\end{equation}
The expectations must therefore both vanish.
Equivalently, tracing either factor in
Eq.~\eqref{eq:supp-singlet} gives $I_d/d$, and hence
\begin{equation}
 \langle\Psi_0|(A\otimes I)|\Psi_0\rangle=\frac{\tr A}{d}
\end{equation}
for every operator $A$. This is the same trace identity used
with a maximally entangled input in
Theorem~\ref{thm:diamond_distance}.
It explains the noise-distance calibration through tracelessness;
the cross-term cancellation instead uses the adjoint identity
in Eq.~\eqref{eq:spin-triple}.

\section{Supplementary covariance calculation}
\label{app:supplementary-covariance}

This appendix expands the Kraus-sum verification already given
after Eq.~\eqref{eq:NLS-covariance}. It records explicitly the
covariance of the noisy channel and its chosen complement for
readers seeking the intermediate steps. It introduces no
additional symmetry assumption for Theorem~\ref{thm:capacity}.

\subsection{The channel output}

Let $U_g$ be the spin-$j$ representation of $g\in SU(2)$.
The generators transform as
\begin{equation}\label{eq:supp-generator-rotation}
 U_g^\dagger J_kU_g=\sum_{l=1}^3R_{kl}(g)J_l,
 \qquad R(g)^{\mathsf T}R(g)=I_3,
\end{equation}
where $R(g)$ is a real rotation matrix. We retain the
normalization $\kappa=j(j+1)$ throughout:
\begin{equation}
 \LL_j(X)=\frac1\kappa\sum_{k=1}^3J_kXJ_k.
\end{equation}
Inserting $U_gU_g^\dagger=I$ on both outer sides of the
Kraus sum yields
\begin{align}
 &\LL_j(U_gXU_g^\dagger)\nonumber\\
 &=\frac1\kappa U_g
   \left[\sum_k(U_g^\dagger J_kU_g)X
                    (U_g^\dagger J_kU_g)\right]U_g^\dagger
   \nonumber\\
 &=\frac1\kappa U_g
   \left[\sum_{l,m}\left(\sum_kR_{kl}(g)R_{km}(g)\right)
                         J_lXJ_m\right]U_g^\dagger
   \nonumber\\
 &=\frac1\kappa U_g\left(\sum_lJ_lXJ_l\right)U_g^\dagger
   =U_g\LL_j(X)U_g^\dagger.
 \label{eq:supp-noise-covariance}
\end{align}
The third line reduces to the fourth because
\begin{equation}
 \sum_kR_{kl}(g)R_{km}(g)
 =(R(g)^{\mathsf T}R(g))_{lm}=\delta_{lm}.
\end{equation}
The identity channel is covariant under the same representation.
Taking the mixture $(1-p)\id+p\LL_j$ therefore gives
\begin{equation}
 \M_{j,p}(U_gXU_g^\dagger)
 =U_g\M_{j,p}(X)U_g^\dagger.
\end{equation}

\subsection{The complementary output}

For completeness, set $V_g=1\oplus R(g)$ in the environment
basis $(0,x,y,z)$. For the Kraus list in
Eq.~\eqref{eq:NLS-kraus},
\begin{equation}
 U_g^\dagger K_\mu U_g
 =\sum_{\alpha=0}^3(V_g)_{\mu\alpha}K_\alpha.
\end{equation}
The complementary-channel convention then gives
\begin{align}
 &[\M_{j,p}^c(U_gXU_g^\dagger)]_{\mu\nu}\nonumber\\
 &=\tr\!\left[(U_g^\dagger K_\mu U_g)X
                  (U_g^\dagger K_\nu^\dagger U_g)\right]
   \nonumber\\
 &=\sum_{\alpha,\beta}
   (V_g)_{\mu\alpha}\overline{(V_g)_{\nu\beta}}\,
        \tr(K_\alpha X K_\beta^\dagger).
\end{align}
Equivalently,
\begin{equation}
 \M_{j,p}^c(U_gXU_g^\dagger)
 =V_g\M_{j,p}^c(X)V_g^\dagger.
\end{equation}
The representation $V_g$ is independent of the noise parameter.
It therefore also applies to $\M_{j,s}^c$ and to
$\M_{j,s}^c\circ\M_{j,p}$, making explicit the entropy
invariance used in Sec.~\ref{sec:capacity}.
% END INLINE SUPPLEMENTARY APPENDICES

% BEGIN INLINE BIBLIOGRAPHY
% END INLINE BIBLIOGRAPHY
\end{document}